%% file: pact22-template.tex
\documentclass[sigconf, anonymous=false]{acmart}

\usepackage{algorithmic}
\usepackage{graphicx}
\usepackage{textcomp}
\usepackage{xcolor}
\usepackage{fancyhdr}
\usepackage{multicol}
\usepackage[ruled,vlined]{algorithm2e}

\theoremstyle{definition}
\newtheorem{definition}{\textbf{Definition}}[section]

\theoremstyle{remark}

\hypersetup{colorlinks,linkcolor={blue},citecolor={blue},urlcolor={red}}

\begin{document}

\title{Com-CAS: Effective Cache Apportioning Under Compiler Guidance}

\author{Bodhisatwa Chatterjee}
\email{bodhi@gatech.edu}
\affiliation{%
  \institution{Georgia Institute of Technology}
  \city{Atlanta}
  \state{GA}
  \country{USA}
}

\author{Sharjeel Khan}
\email{smkhan@gatech.edu}
\affiliation{%
  \institution{Georgia Institute of Technology}
  \city{Atlanta}
  \state{GA}
  \country{USA}
}

\author{Santosh Pande}
\email{santosh.pande@cc.gatech.edu}
\affiliation{%
  \institution{Georgia Institute of Technology}
  \city{Atlanta}
  \state{GA}
  \country{USA}
}

\renewcommand{\shortauthors}{Bodhisatwa Chatterjee, Sharjeel Khan, Santosh Pande.}

\begin{abstract}
With a growing number of cores in modern high-performance servers, effective sharing of the last level cache (LLC) is more critical than ever. The primary agenda of such systems is to maximize performance by efficiently supporting multi-tenancy of diverse workloads. However, this could be particularly challenging to achieve in practice, because modern workloads exhibit \textit{dynamic phase behaviour}, which causes their cache requirements \& sensitivities to vary at finer granularities during execution. Unfortunately, existing systems are oblivious to the application phase behavior, and are unable to detect and react quickly enough to these rapidly changing cache requirements, often incurring significant performance degradation. 

In this paper, we propose \textit{Com-CAS}, a new apportioning system that provides dynamic cache allocations for co-executing applications. \textit{Com-CAS} differs from the existing cache partitioning systems by adapting to the dynamic cache requirements of applications just-in-time, as opposed to reacting, without any hardware modifications. The front-end of \textit{Com-CAS} consists of compiler-analysis equipped with machine learning mechanisms to predict cache requirements, while the back-end consists of proactive scheduler that dynamically apportions LLC amongst co-executing applications leveraging Intel Cache Allocation Technology (CAT). \textit{Com-CAS}'s partitioning scheme utilizes the compiler-generated information across finer granularities to predict the rapidly changing dynamic application behaviors, while simultaneously maintaining data locality. Our experiments show that \textit{Com-CAS} improves average weighted throughput by \textbf{15\%} over unpartitioned cache system, and outperforms state-of-the-art partitioning system \textit{KPart}~\cite{el2018kpart} by \textbf{20\%}, while maintaining the worst individual application completion time degradation to meet various Service-Level Agreement (SLA) requirements.

\end{abstract}

\maketitle

\input{texes/intro}
\input{texes/motivation}

\input{texes/probe}

\input{texes/bcache}

\input{texes/results}
\input{texes/background}
\input{texes/conclusion}

\bibliographystyle{ACM-Reference-Format}
\bibliography{sample-base}


\end{document}

%% file: texes/intro.tex
\section{Introduction}
Modern high performance systems facilitate concurrent execution of multiple applications by sharing resources among them. The \textbf{Last-Level Cache} (LLC) is one such resource, which is usually shared by all running applications in the system. However, sharing LLC often results in \textit{inter-application interference} \cite{temam1994cache, wu2011characterization}, where accesses from multiple applications can map on to the same cache line and incur conflict misses, resulting in performance degradation. This problem is futher exacerbated by the fact that the number of cores, and the processes, that share the LLC, are rapidly increasing in the recent architectural designs. To address this issue, modern processors use \textbf{Cache Partitioning} to divide the LLC among the co-executing applications in the system. The main goal of cache partitioning schemes is to obtain performance isolation to secure dedicated regions of cache memory to high-priority cache-intensive applications to ensure their performance. Apart from achieving \textit{superior application performance} and \textit{improving system throughput} \cite{molnos2006throughput, muralidhara2010intra, el2018kpart, pons2020phase, wang2017swap, sanchez2011vantage, xu2018dcat, yuan2021don, noll2018accelerating, su2018scp}, cache partitioning can also improve \textit{system utilization, power and energy consumption} \cite{ravindran2007compiler, cook2013hardware, kasture2015rubik, nikas2019dicer, shen2019reuse, nejat2020coordinated, huang2018curve, xiang2018dcaps}, ensure \textit{fair resource allocation} \cite{sari2019fairsdp, selfa2017application, zhang2020hica, park2019copart, xiao2020novel, saez2021lfoc+, roy2021satori} and even enable \textit{worst case execution-time analysis} \cite{mueller1995compiler}. Due to such significant benefits of cache partitioning, modern processor families (Intel{\textregistered} Xeon Series) implement hardware way-partitioning through \textbf{Cache Allocation Technology} (CAT) \cite{herdrich2016cache}, which provides the users with extended control and flexibility to customize cache partitions as per their requirements.

On the other hand, modern workloads exhibit \textbf{dynamic phase behaviour} \cite{fried2020caladan, iorgulescu2018perfiso} due to which their resource requirements rapidly change throughout their execution cycle. These changes in phase behaviours arise from the complex control flows, diverse and irregular memory referencing behaviors in the applications, and different sizes of the problems tackled by the loops, which lead to variant footprints and different reuse patterns. As a result, it is often the case that a single program region (such as a given loop in a program) can exhibit variant data reuse behaviour during multiple invocations. Prior works \cite{cook2013hardware, selfa2017application, wang2017swap, xu2018dcat, nikas2019dicer, suo2009system, moreto2007explaining, liang2017exploring, pons2020phase} on cache partitioning tend to classify workloads into categories based on different levels of cache sensitivity. Based on this offline workload characterization, cache allocations are then decided though a mix of static and runtime policies. The fact that a single application can exhibit \textit{dual behaviour}, i.e both cache-sensitive \& cache-insensitive behaviour during its execution is not taken into account by such schemes. This often leads to cache being injudiciously partitioned, and causes performance degradation (\S\ref{sec:overview}). Such static schemes do not account for dynamic loop characteristics dependent on variant loop bounds, memory footprints and reuse behaviors.

Some cache partitioning approaches \cite{cook2013hardware, qureshi2006utility, wang2016secdcp, suh2004dynamic, ye2014coloris, gupta2015spatial} adapt a \textit{`damage-control'} strategy, where the partitions are dynamically adjusted \textit{after} detecting changes in application behaviour using runtime monitoring. Unfortunately, these approaches are \textit{reactive}, and suffer from the problem of \textit{detection and reaction lag}, i.e. they detect the phase change only after a certain interval after its occurrence, and the application behaviour of workloads are likely to change before the adjustments are made, leading to lost performance (\S \ref{sec:background}).

\textit{In this work, we study the problem of providing smart cache partitioning by accounting for applications' dynamic phase behaviors, just-in-time {\it before entering such cache-sensitive program regions}}. Our end goal is to obtain superior performance \& maximal isolation by reducing inter-application interference in multi-execution environments. We propose \textbf{Compiler-Guided Cache Apportioning System (Com-CAS)}, which estimates the dynamic cache requirements for all co-executing applications just-in-time, and then uses this information to perform runtime cache partitions.

\textbf{Probes Compiler Framework} (\S\ref{sec:prob}): \textit{The first contribution of this paper is a compiler framework that estimates applications' cache requirements just-in-time, through a combination of static \& dynamic program attributes} (Algorithm \ref{a3}). Program attributes (\textit{memory footprint, data reuse behaviour, phase timing, cache sensitivity}) that dictate the cache requirements are obtained via traditional compiler analysis coupled with machine learning algorithms. The attributes are then encapsulated in specialized library markers called \textbf{`probes'}, which are statically instrumented outside every loop-nest in the application. The probes communicate these attributes during program execution for cache management. This represents the front-end of \textit{Com-CAS}. 

\textbf{BCache Allocation Framework} (\S\ref{sec:bcache}): \textit{The second contribution of this paper is a new runtime apportioning mechanism, consisting of phase-aware allocation algorithms for determining LLC partitions, based on aggregation of probes-generated information}. The dynamic apportioning mechanism of BCache Framework first detects phase overlaps of co-executing applications and then minimizes application interference. It does so  by avoiding co-location of processes with high data-reuse, and by grouping processes with minimal phase overlap, while maintaining data locality. A net result of such apportioning decisions making is that LLC misses are shifted from applications with higher cache sensitivity (high sensitivity reuse) needs to the lower ones (lower sensitivity reuse or streaming) (\S \ref{res:cache}).

\textbf{Proactive Scheduling}: \textit{The third contribution of this paper is an user-level scheduler that interacts with Intel CAT to perform the actual cache partitioning according to application phase changes.} The phase change occurs whenever an executing application changes its reuse-behaviour or exhibits a significant variation in its memory footprint, which necessitates cache re-allocation. The phase information is relayed by probes to the scheduler, which in-turn invokes apportioning algorithms to obtain phase-based cache partitions for co-executing applications. Anecdotally, compiler-driven analyses usually focus on one application at a time (including locality-promoting loop transformations) to improve single application performance in isolation, and are oblivious to other co-executing applications in a system.  Designing an user-level scheduler allows us to aggregate compiler-generated attributed across all applications, thus making the way for system-wide decision making. To the best of our knowledge, this is the first compiler approach that is able to take a "global" view of the cache sensitivity problem in a shared setting with good performance gains over the current state-of-the-art \cite{el2018kpart}.      

We evaluated \textit{Com-CAS} on an Intel Xeon Gold system with 35 application mixes from a diverse sets of workloads \cite{beamer2015gap, che2009rodinia, pouchet2012polybench} that are typically encountered in a multi-tenant environments (\S\ref{sec:res}). Our target applications \footnote{The current evaluation is throughput-oriented, and adheres to a fixed latency degradation limit. Currently, we do not focus on latency-critical applications.} are loop-heavy, data driven, memory-bound and cache-sensitive workloads from domains such as machine learning, scientific computations, graph analytics, which typically exhibit phased execution behavior. The average weighted throughput gain obtained over all the mixes is \textbf{15\%} (upto \textbf{35\%}). In all the instances, no degradation was observed over a vanilla co-execution environment (with no cache apportioning scheme). In addition, we show that our proactive cache partitioning scheme outperforms state-of-art apportioning system \textit{KPart} \cite{el2018kpart} by \textbf{20}\%, and other hardware counter-based apportioning policies by \textbf{30\%} on average. We also show the \textit{Com-CAS} maximises the isolation of reuse heavy loops reducing inter-application interference simultaneously while ensuring fairness and maintaining individual process latencies. 


%% file: texes/motivation.tex
\section{Background and Motivation}
\label{sec:overview}

Determination of memory requirements and its allocation for applications running in a multi-execution environment is quite a complex task. Typically, in data centers, for maximizing resource utilization, multiple processes execute concurrently in a batch, while sharing the same LLC \cite{amvrosiadis2018diversity}. The cache demand for each process is dependent on the specific program point, as well as the underlying execution phase, and can change from loop nest to loop nest (Fig. \ref{case_study}). Not accounting for these dynamic variations in cache requirements while partitioning the cache can cause significant performance degradation \cite{sanchez2011vantage, yuan2021don}. Thus, a good cache apportioning scheme must take into consideration dynamic concurrency of the appropriate phases of co-executing applications. For example, if loop nest 1 of application 1 is co-executing with loop nest 2 of application 2, the apportioning will be quite different than loop nest 2 of application 1 co-executing loop nest 2 of application 2. None of the current works analyzes and apportions at this level of granularity leading to sub-optimal performance. Moreover, since none of the approaches are compiler-based, they lack the apriori ability needed to {\it predict} the cache behaviors. Table \ref{tab:pol} shows the comparison of four different apportioning policies (\S\ref{sec:res}) operating on a 6-process application mix from \textit{PolyBench}.  As we can observe, apportioning policies that are agnostic to application phases, lead to \textit{injudicious cache partitioning}, causing slowdowns ranging from \textbf{10}\% to \textbf{20}\%, whereas \textit{Com-CAS} leads to a speed-up of \textbf{27.7\%} over the baseline of unpartitioned, unmanaged cache allocation. 

\begin{table}[ht]
\caption{Comparison of various cache allocation policies for a 6-process mix ($Cholesky,Floyd\-Warshall,LudCmp,Lu,Corr,Cov$)}
\resizebox{\columnwidth}{!}{
\begin{tabular}{|l|l|l|l|l|}
\hline
\textbf{Partition Scheme} & \textbf{Policy Type} & \textbf{Interval} & \textbf{Exec Time} & \textbf{Performance} \\ \hline
Max Cache Ways & Static & NA & 1230.41 sec & 19.82\% (Slowdown) \\ \hline
Perf Counter-based & Dynamic (Reactive) & 250 ms & 1234.4 sec & 20.21\% (Slowdown) \\ \hline
KPart & Dynamic (Reactive) & $20^6$cycles & 1119.43 sec & 10.3\% (Slowdown) \\ \hline
Com-CAS & Dynamic (Proactive) & Adaptive & 742.3 sec & 27.7\% (Speedup) \\ \hline
Baseline & Unpartitioned Cache & NA & 1014.39 sec & Baseline \\ \hline
\end{tabular}}
\label{tab:pol}
\end{table}

Modern workloads executing concurrently in multi-tenant environments exhibit variant resource requirements in their nested-loop regions. Leveraging compiler analysis enables us to anticipate these rapid phase changes right at its onset (i.e at the loop entrances), and allows us to adapt the apportioning decisions at a few millisecond granularity. Also, if the differences in cache requirements between successive loop-nests are insignificant, no allocation change is carried out by \textit{Com-CAS} in order to prevent hysteresis. Thus Intel CAT is invoked only when the phase-change entails substantial change in cache requirements.

Finally, our work leverages Intel CAT \cite{herdrich2016cache}, which is a reconfigurable implementation of way-partitioning that allows the user to specify cache capacity designated to each running application. 

\subsection{Intel Cache Allocation Technology (CAT)}
Intel{\textregistered} Cache Allocation Technology (CAT) \cite{herdrich2016cache} allows the user to specify cache capacity designated to each running application. The primary objective of Intel CAT is to improve performance by prioritizing and isolating critical applications. Through isolation, Intel CAT also provides powerful means to stops certain kinds of side channel and DoS attacks. It provides a customized implementation of way-partitioning with a ``software programmable'' user-interface, that is used to invoke in-built libraries to perform cache partitioning. 

To reason about different cache partitions, Intel CAT introduces the notion of \textbf{Class-of-Service} (CLOS), which are distinctive groups for cache allocation. The implication of the CLOS abstraction is that one or more applications belonging to the same CLOS will experience the same cache partition. Intel CAT framework allows users to specify the amount of cache allocated to each CLOS in terms of Boolean vector-like representations called \textbf{Capacity Bitmasks} (CBMs). These are used to specify the entire cache configuration for each CLOS, including isolation and overlapping of ways between one or more CLOS. On the hardware side, Intel CAT uses arrays of MSRs to keep track of CLOS-to-ways mappings. Minor changes are then done in the resource allocation mechanism of Linux Kernel to interact with these dedicated registers. In our work, we use apportioning algorithms to generate CBMs for a process residing in a particular CLOS. We then use specialized library calls to interact with the Intel CAT interface to perform the required allocation.

\textbf{Adaption of Com-CAS for other (Non-Intel) architectures}: In this version, Com-CAS relies on an re-configurable cache partitioning mechanism that allows users to customize the cache partitions. The phase requirement algorithms and the apportioning algorithms are not dependent on any specific architecture as they maximize isolation of cache-sensitive applications. However, the library that interfaces Com-CAS with Intel CAT might need to be re-implemented to work with a new re-configurable cache apportioning system - in case Com-CAS is to be adapted to an entirely new architecture.


%% file: texes/probe.tex
\section{Probes Compiler Framework}
\label{sec:prob}

\begin{figure*}
  \includegraphics[width=1.0\linewidth]{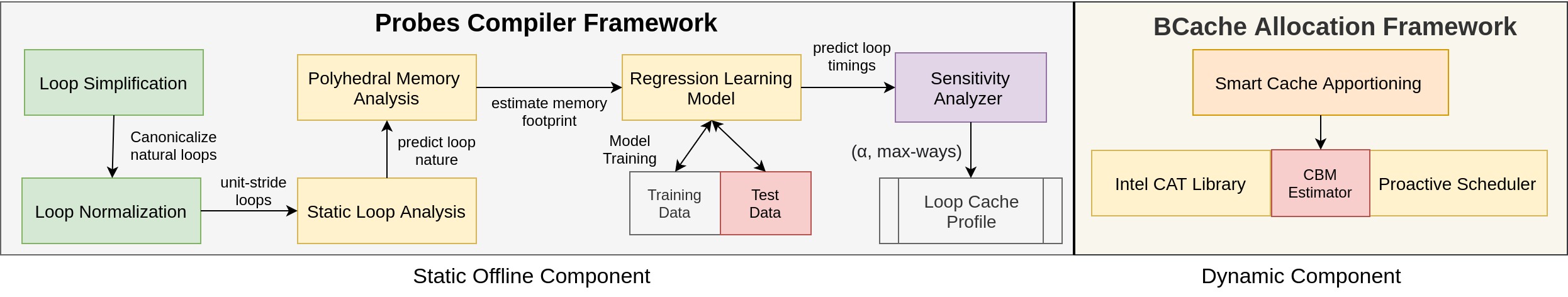}
  \caption{Overview of \textit{Com-CAS} consisting of Probes Compiler Framework (\S \ref{sec:prob}) \& BCache Allocation Framework (\S \ref{sec:bcache})}\label{fig:system}
\end{figure*}

This section describes the front-end (compiler phase) of \textit{Com-CAS} (Fig. \ref{fig:system}). \textbf{Probes Compiler Framework} is a LLVM-based, instrumentation framework equipped with machine learning mechanisms to determine an application's resource requirements across various execution points. It inserts `\textit{probes}' (specialized library markers) at outermost nested-loops within each function of an application. These probes estimate the attributes that dictates an application's cache requirements for each execution phase: \textbf{memory footprint}, \textbf{cache sensitivity}, \textbf{reuse behaviour} \& \textbf{phase timing}. In our work, we consider each nested loop to potentially constitute a new \textit{execution phase}. 

\subsection{Estimating Application Cache Behaviour}

The first three attributes (footprint, sensitivity \& reuse behaviour) quantify the cache requirement during a phase, while loop-timing determines how long a particular phase will last. The memory footprint and reuse behaviour can be statically analyzed as closed-form expressions, while the phase timing must be predicted using learning regression models which can be evaluated at runtime using dynamic values. Such models are statically inserted before the loop nest.  In order to incorporate both the trained model \& compute attributes dynamically, the Probes Framework has two components: \textit{compilation component} and \textit{runtime component}.

\begin{figure}[htp]
\centerline{\includegraphics[width=0.9\linewidth]{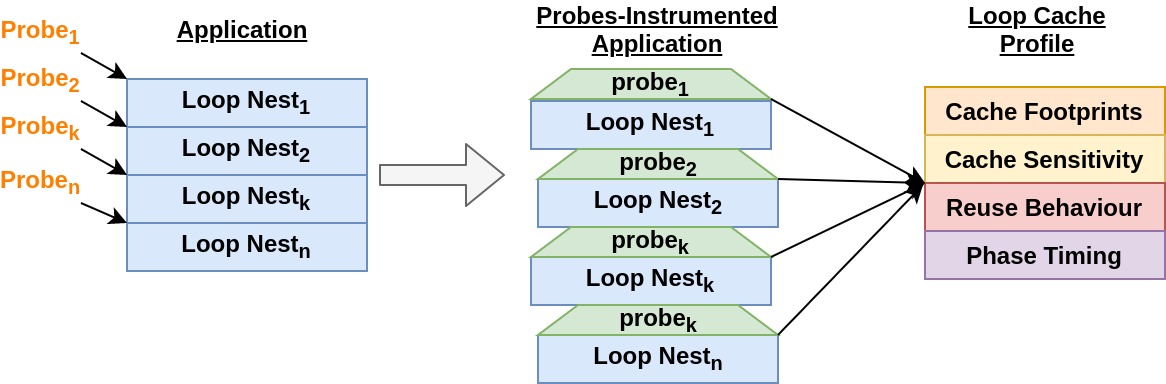}}
\vspace{-0.1in}
\caption{Cache profile for each loop nests are generated \& encapsulated in Probes}
\label{fig:prob}
\end{figure}

The \textit{compilation component} primarily consists of multiple LLVM \cite{lattner2004llvm} compiler passes which instrument probes into application source code and embeds loop memory footprint usage, data-reuse behaviour analysis and trained phase-timing models. Apart from encapsulating loop attributes, the compilation component also inserts a special probe-start function in the preheader of the loop nest and a probe-completion function at the exit node of the loop nest. For loops present in the inner-nesting levels, the probe functions are hoisted outside the outermost loop inter-procedurally \footnote{Please refer to appendix \S C for discussion on recursive calls}. During hoisting, the pass combines all attributes of inner-loops with the outermost loop. For example, if any of the inner-most loop exhibits significant data reuse, then the entire loop-nest is considered to be a reuse-heavy phase. \textit{The runtime component} on the other hand, complements the compilation component by dynamically computing the values of memory footprint usage \& phase-timing and conveying them to the BCache Framework. This communication is facilitated by passing the attributes as arguments to the probe library function calls, which are further transferred to the scheduler via shared memory. 

\begin{algorithm}[htbp]
\footnotesize
\SetAlgoLined
\textbf{Input:} Application $A$\\
\KwResult{Estimate the loop timing, cache footprint, data-reuse behaviour, cache sensitivity}
\For{each loop-nest $L \in A$}
{
\textbf{****** Learning Loop Timing ******}\\
L.normalize()\\
$U \leftarrow getAllUpperBounds(L)$\\
$W \leftarrow intializeWeights()$\\
Generate Training \& Testing Dataset (T, W) with representative inputs\\
Learn the linear regression equation $T_W(u) = W^T U$\\
\textbf{****** Estimating Loop Memory Footprint ******}\\
$N \leftarrow obtainLoopTripCount()$\\
Analyze affine memory-accesses $m(X),~ \forall X \in \{0, N \}$\\
Obtain polyhedral mapping $[X] \rightarrow \{ m(X) : 0 < X < N \}$\\
\textbf{****** Characterizing Loop Data Reuse Behaviour ******}\\
$srd_L \leftarrow calcSRD(L)$\\
\eIf{$srd_L > \Delta$}{$L.type = reuse$}{$L.type = stream$}
}
\textbf{****** Analyzing Application Cache Sensitivity ******}\\
$W_T \leftarrow getTotalCacheWays()$ \\ 
$c \leftarrow calcMaxWays(A)$\\
\For{$i \in \{ 3, W_{max}\}$}
{Generate application timings $t_i, t_{i-1}$ with representative inputs\\
$\Delta t_i = \left | t_{i} -t_{i-1} \right |$\\
$\Delta w_i = \left | w_i -w_{i-1}  \right |$\\
$\alpha_A = \sum_{i=3}^{\max{ways}} \frac{\left |\Delta t_i  \right |}{\Delta w_i}$}
 \caption{Phase Cache Requirement Estimation}\label{a3}
\end{algorithm}
The entire cache requirement estimation process is summarized in Algorithm \ref{a3}. We will now describe in detail how each steps of this algorithm computes the program attributes.

\subsubsection{\textbf{Phase Timing Model}}
\label{prob:time}
The phase-timing is defined as the time taken for executing an entire loop-nest. Probes Compiler Framework uses a linear regression model to predict the execution time for each loop-nest. The underlying idea is to establish a numerical linear relation between loop-timing and loop-iterations. 

\begin{theorem}[Phase Timing] \label{th1}
For a normalized loop nest L with $n$ inner-nested loops with individual upper-bounds $\{U_1,U_2,...,U_n\}$, the timing $T_c$ is given by the linear equation:
\begin{equation}\label{equation:lm}
T_c = \mathbf{U}^T \mathbf{C} = c_0 + c_1 u_1 + c_2 u_2+...+ c_n u_n
\end{equation}
where $C = \{c_0, c_1, .. c_n\}$ are learnable parameters, feature vector $\mathbf{U} = \{u_0,u_1,..,u_n\}$ and feature $u_i$ is mapped from loop upper-bounds $U_i$ as $u_i = \prod_{k = 1}^{i} U_k$
\end{theorem}
\begin{proof}
 In general, for a loop with arbitrary statements in its body, the \textit{loop-timing} is proportional to the loop bounds. Thus, for a loop $L$, with time $T$ and Bounds $B$, we have:
 \begin{gather}
      T \propto B\\ 
      \Rightarrow T = \beta * B
 \end{gather}
 For nested-loops, bounds of each inner-loops must be included as well. Thus, we have:
  \begin{gather}\label{eq:b1}
     T =  f(B_1, B_2, ..., B_n)
 \end{gather}
 where $B_i$ is the loop bound for the $i^{th}$ inner-nested loop. 
 To make this analysis easier, Probes Framework uses the LLVM's \cite{lattner2004llvm} \textit{loop-simplify} \& \textit{loop-normalization} passes, which transforms each loop to lower-bound of 0 and unit-step size, and upper-bound ($u_i$).
 Thus, eq. \ref{eq:b1} becomes:
 \begin{equation}
    T =  f(u_1, u_2, ..., u_n) 
 \end{equation}
 Furthermore, to enable analysis of imperfectly nested-loops, loop distribution can be performed to transform them into a series of perfectly-nested loops. Thus, the phase-timing equation can be decomposed sum of individual $n$ distributed loops timings ($U_i$ represents loop upper bounds): 
\begin{equation}\label{equation:loopLB1}
    T = f_1(U_1) +  f_2(U_1, U_2) ... + f_n(U_1, U_2, ..., U_n)
\end{equation}
In this work, we assume a linear relationship between loop bounds and loop timing (any higher degree polynomial lead to overfitting). Thus, Eq. \ref{equation:loopLB1} can be further mapped into features for a linear regression model, i.e $U_1 = u_1$, $U_1*U_2 = u_2$ and so on. Rewritting, we obtain:
\begin{equation}
    T_c = c_0 + c_1 u_1 + c_2 u_2 +...+ c_n u_n
\end{equation}
where $u_i = (U_1 * U_2 *..* U_i) =  \prod_{k = 1}^{i} U_k$. This creates a mapping from the individual loop bound $U_i$ to the individual feature $u_i$
\end{proof}

Eq. \ref{equation:lm} can be interpreted as a linear-regression model $T_c(u)$ with weights $c_1,c_2,...c_n$ \& intercept $c_0$, where the weights can be learned. Probes framework uses a compiler pass to generate loop-bounds \& the timing for a subset of application inputs to generate training data for the regression model. The training data comprises of multiple program inputs, which the learned co-efficients $c_0, c_1,c_2,...c_n$ generalizable across different inputs. Once the model is trained, the timing equation for each loop-nest is embedded in its corresponding probe. During runtime, the actual loop-bounds are plugged into this phase-timing equation to generate the phase-time. For non-affine loops that have unknown bounds, we generate an approximate loop-bound based on the test input sets to predict the timing.

\subsubsection{\textbf{Loop Memory Footprint}}
\label{prob:foot}
\textit{Memory footprint} determines the amount of cache that will be utilized by an application during an execution phase. Probes framework calculates the memory footprint through static \textit{polyhedral analysis} \cite{grosser2011polly}, which is a standard compiler technique for analyzing the memory accesses in a loop nest. 

Polyhedral analysis generates memory footprint equations of the form: 
\begin{equation}
    [X] \rightarrow \{ m(X) : 0 < X < N \}
\end{equation} 
which represents a mapping from the loop iterations to the dynamic memory accesses by a statement and $N$ is the expected iterations of the loop nest. Therefore, at runtime with the value of loop upper-bound $N$, probes calculates the exact footprint \& passes it to the runtime component.

\subsubsection{\textbf{Classifying Reuse Behaviour}}
\label{prob:reuse}
Most applications exhibit temporal and spacial locality across various loop iterations. To fully utilize the locality benefits, the blocks of memory with reuse potential must remain in the cache, along with all the intermediate memory accesses. To determine the amount of cache required by a loop-nest for maximizing locality, we need to obtain a sense of reuse behaviour exhibited by the loop-nest. To classify reuse behaviour, Probes Framework defines \textbf{Static Reuse Distance (SRD)} as follows: 
\begin{definition}[Static Reuse Distance (SRD)]
\textit{The Static Reuse Distance between program statements $S_1$ and $S_2$ is a measure of the number of memory instructions between two accesses to the same memory location from $S_1$ and $S_2$.} 
\end{definition}
Based on this metric, the loop-nests could either be classified as \textbf{Streaming} - if its SRD is negligible, or else \textbf{Reuse} - if the SRD is significant and the reuse occurs over a large set of intermediate memory instructions. Consequently, reuse loops require significantly larger cache than streaming loop and this has to be accounted while deciding cache apportioning. Typically in such cases, the reuse is carried by the outer loop and a large number of inner loop iterations must be completed before the reuse occurs thus needing a larger cache size to hold the outer loop references. Appendix \S C shows an example of how SRD can be leveraged to classify data-reuse behaviour of loop nests.

At runtime, the exact SRD value is computed dynamically and passed on to the probe functions. Similar to footprint estimation, approximate bounds are generated to compute SRD for loops whose bounds are unknown. Also, loops containing indirect subscripts or unanalyzable access patterns are classified as \textit{reuse} since at runtime they could potentially exhibit a data reuse.
\subsubsection{\textbf{Quantifying Cache Sensitivity}}
\label{sec:max}
Intuitively, the performance of a cache-sensitive application should improve with increase in allocated cache size. However, most applications often experience performance saturation - after a certain number of allocated cache-ways, there is no further performance benefit with more cache allocation. Furthermore, before the performance saturation, the more cache-sensitive an application is, the greater performance benefits it experiences with each increment of cache allocation. In this work, we quantify performance saturation and cache sensitivities by the following metrics:

\begin{definition}[Max-Ways ($W_{max}$)]
\textit{The number of cache ways that correspond to performance saturation point of an application, is defined as \textbf{max ways}. For applications that are cache-insensitive, the max-ways is assumed to be two, as allocating less than two ways degrades performance by behaving as a directly-mapped cache. 
}
\end{definition}

\begin{definition}[Performance Sensitivity Factor ($\alpha$)]
\textit{For an application $A$, with loop nest execution times $t$, and allocated cache ways $w$, the \text{performance sensitivity factor} is defined as: 
\begin{equation}
\alpha_A = \sum_{i=3}^{\max{ways}} \frac{\left |\Delta t_i  \right |}{\Delta w_i}    
\end{equation}
where $\Delta t_i = \left | t_{i} -t_{i-1} \right |$ and $\Delta w_i = \left | w_i -w_{i-1}  \right |$, while $i$ denotes an observed data-point.  A high value of $\alpha$ indicates that the application's performance benefits could be increased by increasing the cache allocation when a particular loop nest is being executed.}
\end{definition}
In other words, $\alpha$ captures the change in an application's loop nest execution times as a function of cache ways allocated to that application. Probes profiles and sends the tuple $(\alpha$, \textit{max-ways}) for an application to the scheduler at runtime to guide the apportioning algorithms. In case we have data-points for all consecutive cache-ways (2,3,4,$\dots$n), $\Delta w_i = 1$.

%% file: texes/bcache.tex
\section{BCache Allocation Framework}
\label{sec:bcache}

\textbf{BCache Allocation Framework} obtains efficient cache partitions for  diverse application mixes, based on their execution phases. Typically, during a particular instance of system execution, an application executing a reuse loop with a higher value of memory footprint is allocated a greater portion of the LLC, compared to another application executing a streaming loop. The phase change occurs whenever an executing application changes its reuse-behaviour or exhibits a significant variation in its memory footprint, necessitating cache re-allocation. The phase information is relayed by respective probes to the scheduler, which in turn, invokes apportioning algorithms. 

For grouping different applications into cache partitions, Intel CAT~\cite{herdrich2016cache} introduces the notion of \textbf{Class-of-Service} (CLOS). Applications grouped in the same CLOS share the same cache partitions, and therefore must have similar cache attributes, and must be compatible (should not lead to conflict misses). 

Each CLOS can be mapped to a specific cache configuration. These cache configurations are expressed in terms of \textit{Capacity Bitmasks} (CBM), which is Boolean array of length $n$, where $n$ represents the total number of cache ways present in the system. The CBMs serve as input to Intel CAT for the actual partitioning.  $CLOS \rightarrow CBM$ mapping is obtained by setting the appropriate bits in the array. 

When the processes start executing, they are assigned to a suitable CLOS based on the triggered probe attributes ascertaining their compatibility. As the execution proceeds, the underlying cache configurations for each CLOS are adjusted as per the needs of their constituent processes, while maintaining data-locality. We will now describe our apportioning scheme and the allocation algorithms in detail.

\subsection{Cache Apportioning Scheme}
\label{bcache:scheme}
BCache allocation framework determines cache apportions for an application based on its loop memory footprint \& data-reuse behaviour by adopting an \textit{unit-based fractional} partitioning scheme. The underlying idea of this scheme is that each application will be allocated a fraction of the LLC, which is measured by estimating how much they contribute to the entire memory footprint. 

The memory footprint of each loop is first scaled as per whether it is a streaming loop or a reuse loop. Scaling ensures that the reuse loops get a bigger portion of cache than streaming loops. For an application $K$ with current executing loop $n$ having memory footprints $m$, the adjusted loop memory footprint is defined as: $^Km_{phase} = {}^Km * S_{n}$,
where $S_n$ denotes the \textit{scaling factor} of the current loop $n$. It can be adjusted dynamically and depending on the reuse behaviour exhibited by a loop, \textit{scaling factor} ranges from: 
\begin{equation}
S_n = \left\{\begin{matrix}
= 1, &\text{if n}\rightarrow \text{reuse} \\ 
< 1,  &\text{if n}\rightarrow \text{stream} 
\end{matrix}\right.    
\end{equation}
Based on the adjusted footprints values of all executing loops in the system, a fraction of the cache size that will be allocated to each application is determined. This fraction is calculated by determining how much does the current application's loop footprints contribute to the overall loop footprints in the system. In a system with $N$ applications, the fraction of cache allocated to application $K$ at time $t$ is given by: 
\begin{equation}
f^{cache}_t(K) = \frac{^Km_{phase}}{\sum_{I=1}^{N} {}^Im_{phase^*}}
\end{equation}
where $^Km_{phase}$ denotes the adjusted loop memory footprint of application $K$ and $^Im_{phase^*}$ denotes the adjusted loop memory footprint of application $I$ in their current execution phases. At $t=0$, all the applications will be allocated an initial fraction of cache according to the memory footprints of their first executing loop and the sum of these fractions for all the applications at $t=0$ will be equal to 1. This means that unless applications are grouped in same CLOS, each running application will start their execution with an isolated portion of cache. However, an important observation in the above equation is that each fraction is being calculated based on an individual application's current phase and will be recomputed during their phase changes. Obviously, different co-executing applications in the system might not undergo phase change at the same time as others. This might result in an overlap between cache portions allocated to different applications and the sum of cache fractions will deviate from 1. Based on the sum of cache fractions for all applications in the system, Table \ref{t1} shows the three possible system scenarios at time $t$. 

\begin{table}[ht]
\caption{Possible scenarios in fractional cache partitioning scheme}
\vspace{-6pt}
\resizebox{\columnwidth}{!}{
\begin{tabular}{|l|l|}
\hline
\textbf{\begin{tabular}[c]{@{}l@{}}Sum of Cache fractions\\ $\sum_{I=1}^{N} f^{cache}_t(I)$\end{tabular}} & \textbf{System Scenario} \\ \hline
$= 1$ & \begin{tabular}[c]{@{}l@{}}Cache is fully occupied and each application\\ group has separate portion of cache\end{tabular}     \\ \hline
$> 1$ & \begin{tabular}[c]{@{}l@{}}Cache is fully occupied and some applications\\ have overlapping cache portions\end{tabular}        \\ \hline
$< 1$ & \begin{tabular}[c]{@{}l@{}}Cache is not fully occupied and each\\ application group have separate portion\\of cache\end{tabular} \\ \hline
\end{tabular}}
\label{t1}
\end{table}

\subsection{Phase-Aware Dynamic Cache Allocation}
\label{bcache:algo}
Calculating cache partitions and generating their respective CBMs from cache fractions is achieved by \textit{Initial Phased Cache Allocation} (IPCA) and \textit{Phase Change Cache Allocation} (PCCA) algorithms. The IPCA algorithm (Algo \ref{a1}) is invoked whenever an application begins its execution. A separate socket is dedicated for applications with a higher value of $\alpha$, while low-$\alpha$ applications resides in the other sockets. If there are no cache-ways left in the high-$\alpha$ socket, then applications are grouped in different sockets according to the available cores in the sockets.

\begin{algorithm}[ht]
\footnotesize
\SetAlgoLined
\textbf{Input:} Process P\\
\KwResult{Find effective cache partition for an application based on its initial phase}
\textbf{****** Initializing Socket ******}\\
\eIf{$(P\rightarrow \alpha) > 1$}
{\If{high-$\alpha Socket.ways() > (P\rightarrow \max ways)$}
{Assign P to high-$\alpha$ Socket}
}
{{Assign P to a different socket with max available cores}}
\textbf{****** Selecting CLOS ******}\\
\eIf{$socket.availCLOS > (0.75 * socket.totalCLOS)$}
{find new clos for REUSE process\\
find compatible clos for STREAM process\\}
{compatible clos $=$ all clos where $total\_process \leq GFactor$\\
group STREAM process in compatible clos having $\max(\Delta_{t})$ value\\
group REUSE process in compatible clos having $\min (\frac{\alpha}{\Delta_{t}})$}
\textbf{****** Allocating cache ways ******}\\
$req\_ways = P\rightarrow getCacheApportion()$;\\
\eIf{$avail\_ways > req\_ways$}
{Allocate $req\_ways$ to $P\rightarrow CLOS$}
{Allocate $avail\_ways$ to $P\rightarrow CLOS$}
Generate appropriate bitmasks for allocated ways
 \caption{Initial Phased Cache Allocation (IPCA)}\label{a1}
\end{algorithm}

After socket allocation, a compatible CLOS is obtained for the process. If more than 75\% of total CLOS on a socket are vacant, each reuse processes gets its own CLOS. For processes executing streaming loops, separate CLOS is assigned only if there is no compatible CLOS available. Consequently, when the number of occupied CLOS in a socket crosses the 75\% mark, both streaming and reuse loops are grouped in compatible CLOS groups. Next, the IPCA algorithm estimates the initial cache demand by using the fractional scheme, which is passed into a bitmask generator to obtain the required CBMs after checking the available cache-ways in the system. In case the demand exceeds the available cache capacity, the framework allocates all possible ways to the CLOS and marks it as `unsatisfied'. Finally, Intel CAT is invoked by these generated bitmasks to create partitions.

\begin{algorithm}
\footnotesize
\SetAlgoLined
\textbf{Input:} Process* P\\
\KwResult{Find efficient cache partition for an application based on its phase change}
 $req\_ways = P\rightarrow getCacheApportion()$;\\
 $curr\_ways = P \rightarrow ways$;\\
 \eIf{$req\_ways > curr\_ways$}
 {$extra\_ways = req\_ways - curr\_ways$;\\
 $avail\_ways = socket \rightarrow getavailWays()$;\\
 \eIf{$avail\_ways \geq extra\_ways$}
 {Allocate $extra\_ways$ to $P\rightarrow CLOS$}
 {Allocate $avail\_ways$ to $P\rightarrow CLOS$} 
 }
 {$free\_ways = curr\_ways - req\_ways$;\\
 Allocate $free\_ways$ to most unsatisfied $P^{'}\rightarrow CLOS$;}
 \caption{Phase-Change Cache Allocation (PCC)}
 \label{a2}
\end{algorithm}

The BCache scheduler is implemented as an user-level process that ‘tasksets’ processes onto particular cores and manages them according to the probe information. The probes conveys the scheduler of possible phase-changes in application and PCCA algorithm (Algo. \ref{a2}) is invoked to update the existing cache partitions based on the new requirements. The PCCA algorithm obtains the new demand by using the same fractional apportioning scheme and checks for changes in cache demand. If demand increases, then extra ways are allocated to the CLOS if possible. However, if the current demand of the application can be satisfied with lesser ways, then extra ways are freed and allocated to the `most unsatisfied CLOS' ($\alpha_{\max}$) in the system. The required CBMs are obtained from the bitmask generator accordingly. A similar approach is followed when a loop finishes its execution, along with updating all system parameters like occupied CLOS, available ways, etc. Overall, the cache allocation algorithms manages the following aspects of the system:
\begin{itemize}
    \item \textbf{Preserving Data Locality}: All processes are confined to their initial socket throughout execution time. Applications that possess reuse loops with large footprints are kept within the same CLOS, to maintain a fixed subset of cache-ways for their execution period. On account of phase changes, the number of ways in that CLOS are expanded or shrunk in a way that ensures that an application doesn't start over with cold cache. 
    \item \textbf{Finding Compatible CLOS}: When total allocated ways reaches above 75\% of the maximum available ways, processes are grouped if their way-demands are similar (same cache configuration) and their difference in execution time ($\Delta t$) is maximum. This ensures that grouped processes won't co-execute for long. This allows aggressive grouping of streaming process. However, for reuse process, the compatibility relation also includes grouping with least sensitive processes ($\min \alpha$) in the system to ensure overall performance will not be impacted.
    \item \textbf{Limiting Total Process in a CLOS}:The allocation algorithms constraints that at any time instant, total processes in a given CLOS $\leq GFactor$ (adjustable parameter) to avoid excessive thrashing. This prevents CLOS over-crowding and limits application interference. 
\end{itemize}

%% file: texes/results.tex
\section{Results \& Evaluation}
\label{sec:res}

We evaluate Com-CAS to answer the following set of questions:
\begin{itemize}
     \item How accurately does Com-CAS predict the phased-cache requirements (\S \ref{res:att}) and how effective are the resulting partitions in terms of maximizing throughput? (\S\ref{res:thr}, \S\ref{sub:SLA})
    \item How does the dynamic proactive apportioning scheme compares with other static/dynamic apportioning mechanisms and state-of-art KPart? (\S \ref{res:thr}, \S \ref{res:cache})
    \item What were the internal apportioning decisions taken by \textit{Com-CAS} that allows it to maximize performance (\S\ref{res:case})?
    \item How are the individual process latencies? Is Com-CAS fair to all co-executing processes in the system? ( \S\ref{sub:SLA}) 
\end{itemize}

\textbf{Experimental Setup}: \textit{Com-CAS} is evaluated on a Dell PowerEdge R440 server with Intel Xeon Gold 5117 processors and way partitioning support through Intel CAT. The system is equipped with 2 sockets (14 cores \& 16 CLOS per socket) and 11-way set-associative, 19 MB shared LLC, running Ubuntu 18.04. Scaling Factor ($S_n$) was set to 0.1 for streaming process and $GFactor$ was set to 4. 


\textbf{Benchmarks}\footnote{Our benchmark selection was influenced by the cache-sensitivity study done in \cite{cook2013hardware}. This shows that popular SPEC CPU \cite{henning2006spec} benchmarks are mostly cache-insensitive, while Graph500 (GAPS) is highly cache-sensitive}: Four$^*$ diverse benchmark suites were chosen for our experiments: \textbf{GAP} Benchmark Suite \cite{beamer2015gap}, \textbf{PolyBench} Suite \cite{pouchet2012polybench}, \& \textbf{Rodinia} \cite{che2009rodinia}. These benchmarks represent memory-reuse heavy workloads from machine learning, data and graph analytics, etc, i.e domains that typically require multi-tenant execution environment \cite{amvrosiadis2018diversity}. Refer to Appendix \S A.1 for justification about benchmark selections and Table \ref{tab:ben} for details about benchmarks used for evaluation. For the sake of completeness, we also included popular SPEC 2017 \cite{henning2006spec} in our evaluations (Appendix \S A.2), inspite of it being mostly cache-insensitive. All the benchmarks were compiled with their default optimization flags (-O3/O2, -fopenmp, etc). 

\begin{table}[htbp]
\caption{List of Benchmarks used for evaluation}
\resizebox{\columnwidth}{!}{
\begin{tabular}{|l|l|l|}
\hline
\textbf{Suite} & \textbf{Benchmarks} & \textbf{Input} \\ \hline
Polybench & \begin{tabular}[c]{@{}l@{}}Lu Correlation Covariance Gemm Symm Syr2k\\ Cholesky Trmm 2mm 3mm Doitgen Floyd-Warshall\\ Fdtd-2d Heat-3d Jacobi-2d, Seidel-2d\\ Nussinov Gramschimdt Syrk Adi Ludcmp\end{tabular} & \begin{tabular}[c]{@{}l@{}}MINI, STANDARD,\\ LARGE (Training)\\ EXTRALARGE \\ (Testing)\end{tabular} \\ \hline
GAP & BC CC CC\_SV TC PR SSSP BFS & \begin{tabular}[c]{@{}l@{}}Uniform Random\\ Graph (Train - $2^{22}$ nodes, \\ Test - $2^{24}$ nodes)\end{tabular} \\ \hline
Rodinia & \begin{tabular}[c]{@{}l@{}}Backprop LU Heartwall CFD  Hotspot Srad \\ Particlefilter Streamcluster\end{tabular} & \begin{tabular}[c]{@{}l@{}}Default Inputs (Train)\\ Customized Inputs (Test)\end{tabular} \\ \hline
\end{tabular}}
\label{tab:ben}
\end{table}

\textbf{Parallel Execution}: \textit{Com-CAS} supports OpenMP Parallelism with cores assigned as specified by OpenMP pragmas and default environment variables. Probes are inserted before parallel-loop pragmas and cache requirements are calculated normally; threads are forked into assigned cores and allocated cache is shared among them. PCCA is triggered only when all threads join. Thus, Com-CAS operates on \textit{aggregate cache requirements of all threads at process-level}. GAPS \& Rodinia (partially) uses OpenMP compilation.

\textbf{Creating Realistic Workload Mixes}: Our goal is to evaluate \textit{Com-CAS} on `groups of applications that execute concurrently' (mix), on a real-world, multi-tenant shared cluster. Therefore, our mix creation process aims to encompass all possible resource demand scenarios (either core count/cache occupancy) that can occur in such a multi-execution setting - \textbf{(a)} mixes which exhibit minimal/no contention for the shared resources (\textbf{light-mixes}), \textbf{(b)} mixes whose resource demand just touches the system thresholds (\textbf{medium-mixes}) and \textbf{(c)} mixes composed of processes that can potentially saturate the system in terms of either cache demand or core demand (\textbf{heavy-mixes}). Thus, we created 35 workload mixes, with special emphasis on \textit{heavy-mixes} (21 mixes) that can stress-test the potential of our allocation algorithms. Apart from that, there are 8 \textit{light-mixes} and 6 \textit{medium-mixes}. The applications with $\alpha > 1$ are added in every mix since we are interested in determining the effectiveness of our framework in sensitive benchmarks.  

\textbf{System setting and metric}: We target batch execution for the mixes, a common setting to run modern workloads pertinent to data analytics, machine learning and graph processing in a multi-tenant setting. We use throughput maximization as the key metric for performance improvement. In addition, we use service level agreement (SLA) as a constraint; SLA degradation should be limited to 15\% in our setting. 

\textbf{Probe Overheads}: The probe-instrumented applications have two sources of overheads: (a) short probe-library calls during information broadcast that are in range of $\sim 10 \mu s$ (less than 1\% overhead) and (b) training the regression models that adds $\sim 120$ secs, and embedding adds $\sim 250$ secs to compilation time. Overall, the size of final instrumented binary is about $15\%$ greater than the original binary.

\begin{figure*}[!ht]
  \centerline{\includegraphics[width=1.0\linewidth]{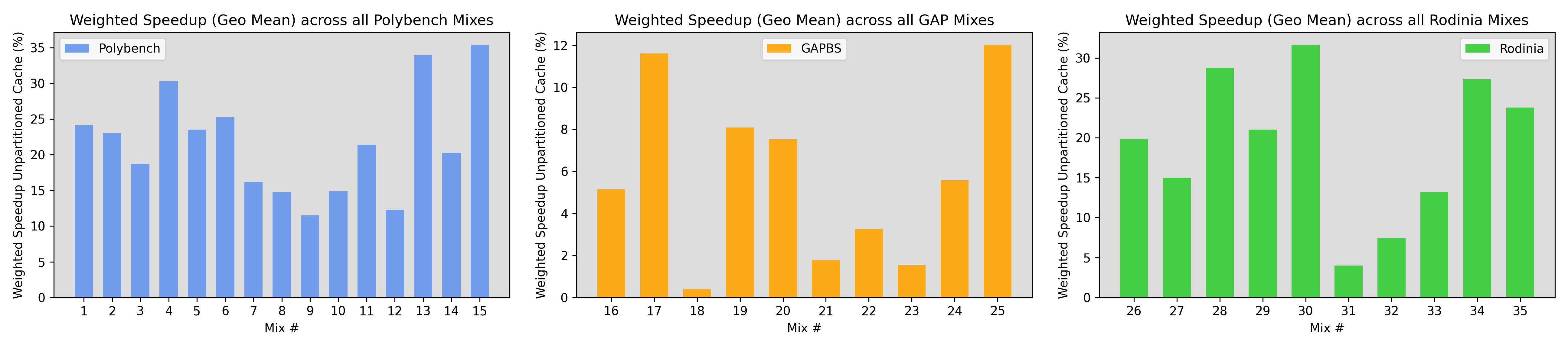}}
  
\vspace{-0.1in}
  \caption{Com-CAS's weighted speedup (15\% Avg) over \textbf{Unpartitioned Cache System} for all 35 mixes. The largest performance gains are in \textit{heavy-mixes}, where the workload resource requirement saturates the system and judiciously partitioning the cache is highly contingent upon utilizing dynamic phase attributes like phase-timing, reuse-behaviour, etc as done by \textit{Com-CAS}.}\label{overall_time}
\end{figure*}

\begin{figure*}[!ht]
  \centerline{\includegraphics[width=1.0\linewidth]{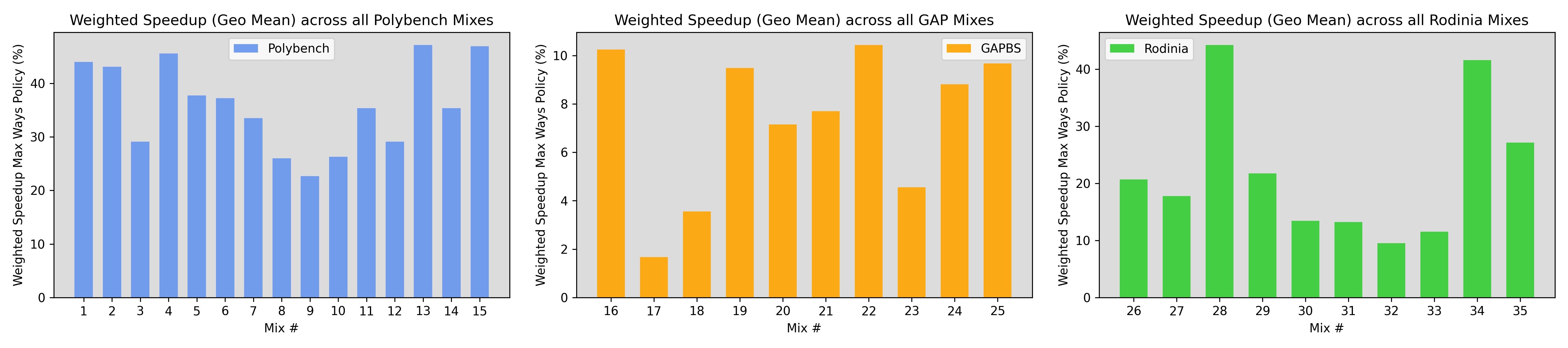}}
  
\vspace{-0.1in}
  \caption{Com-CAS's weighted speedup (21\% Avg) over \textbf{Max-Ways Cache Partitioning Policy} for all 35 mixes. Allocating fixed cache portion that amounts to each application's `max-ways' throughout its execution, results in maximum inter-application interference, where benefits of cache partitioning are squandered.}
  \label{overall_time_sp2}
\end{figure*}

\begin{figure*}[!ht]
  \centerline{\includegraphics[width=1.0\linewidth]{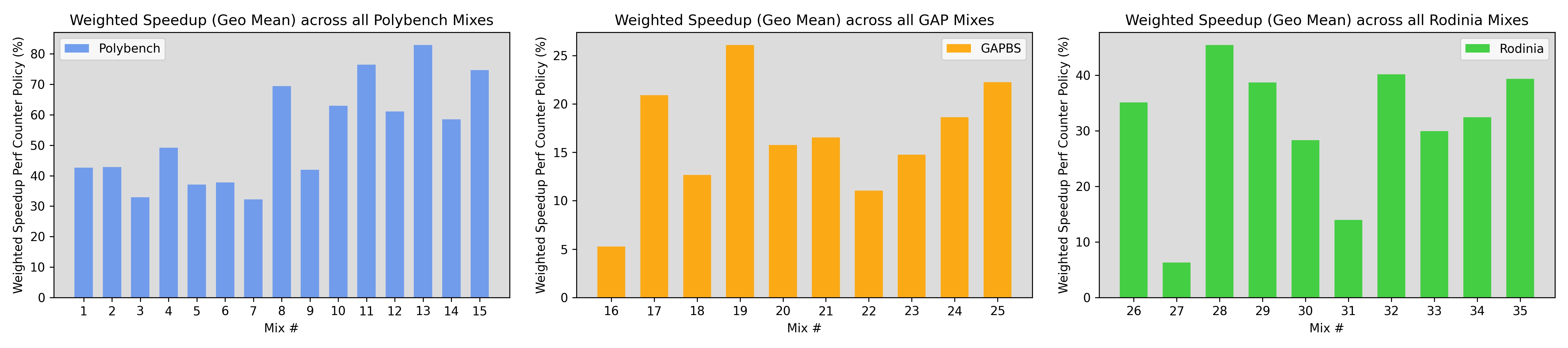}}
  
\vspace{-0.1in}
  \caption{Com-CAS's weighted speedup (32\% Avg) over \textbf{HW Perf Counter-based Cache Partitioning Policy} for all 35 mixes. The strategy of changing the cache-allocations at a fixed-interval (500ms) results in detection lag for applications' phase-changes and execution time worsens for all mixes. This effect can be seen even in \textit{light} \& \textit{medium mixes}, where performance degradation is more pronounced than other policies.}\label{overall_time_res}
\end{figure*}

\begin{figure*}[!ht]
  \centerline{\includegraphics[width=1.0\linewidth]{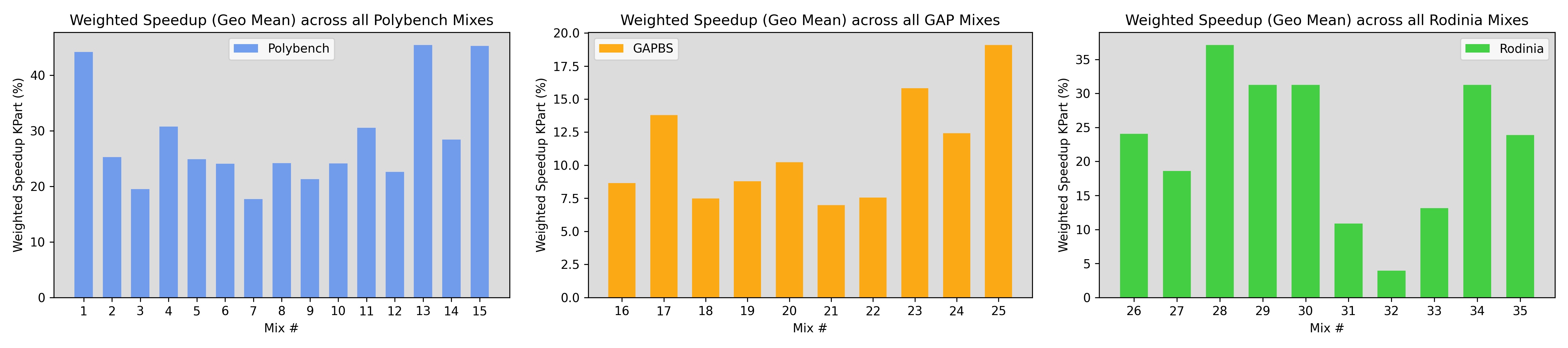}}
  \vspace{-0.1in}
  \caption{Com-CAS's weighted speedup (20\% Avg) over state-of-art \textbf{KPart} for all 35 mixes. \textit{Com-CAS} adjusts the cache allocations during each phase-change, while KPart changes them per $20^6$ cycles. As a result, the performance degradation occurs from non-adaptability to change in application's resource requirements and using fixed samples of cache apportions during each allocation.}\label{overall_time_kpart}
\end{figure*}
\textbf{Comparison with other systems}: We evaluate \textit{Com-CAS}'s performance against four different baselines as follows: 

\begin{itemize}
    \item \textbf{Unpartitioned Cache System}: mixes are simply compiled with normal LLVM and executed on Linux's CFS Scheduler \cite{kobus2009completely}, which is the most commonly used  scheduler in real-world systems with no cache apportioning.
    \item \textbf{Max Ways Static Policy}: Each process in the mix is statically allocated its \textit{max-ways} number of cache ways. The purpose of this policy is to test whether the knowledge of `optimal-cache ways' obtained by offline profiling, is sufficient to perform effective cache partitioning or not.
    \item \textbf{Perf-Counter Based Dynamic Policy}: The objective of this policy is to test whether an application's dynamic phase-behaviour can be predicted by employing hardware counters and then adjusting the cache partitions according to changes in applications' cache-misses and the IPC. This dynamic reactive policy starts off with apportioning the cache equally among all the processes and then at every fixed-time interval (500ms), allocates more cache-ways to processes that exhibits higher $\frac{Cache\-Miss}{IPC}$ ratio.
    \item \textbf{KPart} \cite{el2018kpart}: State-of-art cache partitioning system that groups applications in distinctive clusters by checking the compatibility (reduction in combined cache miss) among applications sharing a cluster by profiling. It then tries to find effective partitions for each cluster and updates them periodically per $20^6$ cycles of instructions. In our evaluation, we minimally modified the KPart code to account for sockets and high core-counts.
    .
\end{itemize}

\subsection{Loop Attribute Prediction Accuracies}
\label{res:att}
We first empirically show that \textit{Com-CAS} is able to accurately predict the loop execution times and thus, cache footprints. The workloads are profiled by the Probes Framework on training inputs (Appendix \S A.1) to obtain loop-timing models. The accuracy of these models are obtained by comparing the predicted values with the actual attribute values on the testing inputs (Table \ref{tab:ben}). The loop timing predicted by the Probes has an \textbf{86\%} average accuracy across all the benchmarks, with lowest accuracy of \textbf{64\%} ($nussinov$) and highest accuracy of \textbf{100\%} ($symm$). Loop timing is affected by the hoisting of probes because it replaces exact values with expected values leading to less precise timing and footprint values. This result can be seen especially in GAP benchmarks because these benchmarks tend to have more interprocedural loops compared to Polybench and Rodinia. Few benchmarks in Polybench ($symm, syr2k$) exhibits typically uniform behaviour across inputs sets, leading to almost perfect timing accuracy ($>95\%$). Overall, high timing accuracies (Appendix \S A.5) show that \textit{Com-CAS} was able to accurately predict the phase timings at fine levels of granularities. 
\vspace{-6pt}

\subsection{Improvement in Weighted Execution Time}
\label{res:thr}
Figs. \ref{overall_time}-\ref{overall_time_kpart} summarizes \textit{Com-CAS}'s weighted geometric mean speedup in terms of achieved throughput on the 35 mixes in each of the four baselines respectively. Compared to the \textbf{Unpartitioned Cache System} (Fig. \ref{overall_time}), \textit{Com-CAS} showed an average throughput improvement of \textbf{21\%} over all 15 \textit{Polybench} mixes, \textbf{5\%} over all 10 GAPS mixes and \textbf{19\%} over all 10 \textit{Rodinia} mixes. 

In particular, we found that \textit{heavy-mixes} from all three benchmarks (mixes \#8 - \#17, \#24 - \#25, \#31 - \#35) showed the largest performance gains. This shows that \textit{Com-CAS} manages to perform effective cache allocation and superior scheduling particularly when the workload resource demands are saturated and it prevents overwhelming of the system. In addition to that, the Unpartitioned Cache System can lead to complete sharing of ways and causing high inter-application interference, while \textit{Com-CAS} side-steps this issue by controlling the number of reuse processes grouped in the same CLOS ($GFactor = 4$). A discussion on \textit{Com-CAS}'s scalability is given in Appendix \S A.4. 

We found that \textbf{static max-ways policy} (Fig. \ref{overall_time_sp2}) performs significantly worse (\textbf{40\%}-\textbf{50\%} slowdown) in PolyBench and Rodinia over the baseline of \textit{Unpartitioned Cache System}. This is because mix heterogeneity \& mix size (specially \textit{medium-mixes} \& \textit{heavy-mixes}), as well as individual process $\alpha$, is higher than GAPS workload. All these factors results in excessive inter-application interference when all processes are allocated `max-ways' number of ways and hence benefits of cache partitioning is squandered. Conversely, \textbf{Perf counter-based reactive scheme} (Fig. \ref{overall_time_res}) that adopts a damage-control methodology at 500 ms granularity, fails to detect the phase changes that vary from few ms to $10^2$ ms and therefore performs significantly worse across all mixes. Interestingly, this policy incurs a greater amount of performance degradation in long-running, \textit{light} \& \textit{medium-mixes} compared to other policies, suggesting that even in workloads containing minimal number of processes (with their collective resource requirement below the system limit), the `detection lag' between an application's actual resource demand and resource allocation by the system results in injudiciously cache partitioning, leading to significant performance degradation. Although this policy can potentially perform better by fine-tuning the fixed-time interval for performing allocations, but this requires a trial-and-error approach and will vary according to the mix composition, which makes it impractical. 

On the other hand, Com-CAS obtains an average improvement of \textbf{44\%} across all the mixes over state-of-the-art \textbf{KPart} (Fig. \ref{overall_time_kpart}). Although \textit{KPart} uses a dynamic partitioning policy and the compatible applications clusters are figured out apirori, the performance gap primarily arises from the non-adaptability of KPart to the varied phase timings (Fig. \ref{case_study}) and inability to judge the appropriate amount of cache allocation for each application. While \textit{Com-CAS} changes the allocations at every loop-nest level with memory footprint variation, \textit{KPart} indiscriminately changes allocations every $20^6$ cycles with different partition samples, which hinders performance. This highlights an important distinction between \textit{KPart} \& \textit{Com-CAS}: \textit{KPart essentially treats each application like a black-box and attempts to find `best cluster fit' for a certain application mix, while Com-CAS is assisted by a compiler framework that analyzes each application separately and guides the cache apportioning at loop nest level}. Thus, the attributes that \textit{Com-CAS} uses to guide cache allocations (loop timing, footprints, reuse behaviour) are better indicators of an application's resource requirement than aggregate measures such as cache-misses and IPC. Similar to the \textit{unpartitioned cache baseline}, \textit{Com-CAS}'s performance enhancement is more pronounced in \textit{heavy-mixes}. 


\subsection{Com-CAS's Process Scalability}
\label{app:scal}
Fig. \ref{mix_time} (right) gives a glimpse into the apportioning decisions taken by \textit{Com-CAS} across all the 35 mixes. The `\textit{total apportions}' for each mix denote the number of partitioning decisions undertaken by \textit{Com-CAS}. Although the number of apportions rises with the total processes in the mix, we can observe that the rate of increase in apportions is not as steep compared to rate of increase in total processes in a mix. This hints at the scalability of \textit{Com-CAS} with respect to mixes containing larger number of processes. Another interesting observation is that maximum processes that are grouped in one CLOS at any point of time, is closer to 4 ($GFactor$), only when the mix size is large or number of applications exhibiting streaming behaviour is dominant. 
Fig. \ref{mix_time} (left) shows total mix execution time for all the four baselines. This represents the execution time for the longest running process in the mix (typically characterized by high $\alpha$) and determines the overall throughput. \textit{Com-CAS} achieves an improvement of \textbf{20\%} (avg) over unpartitioned cache and  \textbf{40\%} (avg) over \textit{KPart}.

\begin{figure*}[htbp]
\centerline{\includegraphics[width=0.8\linewidth]{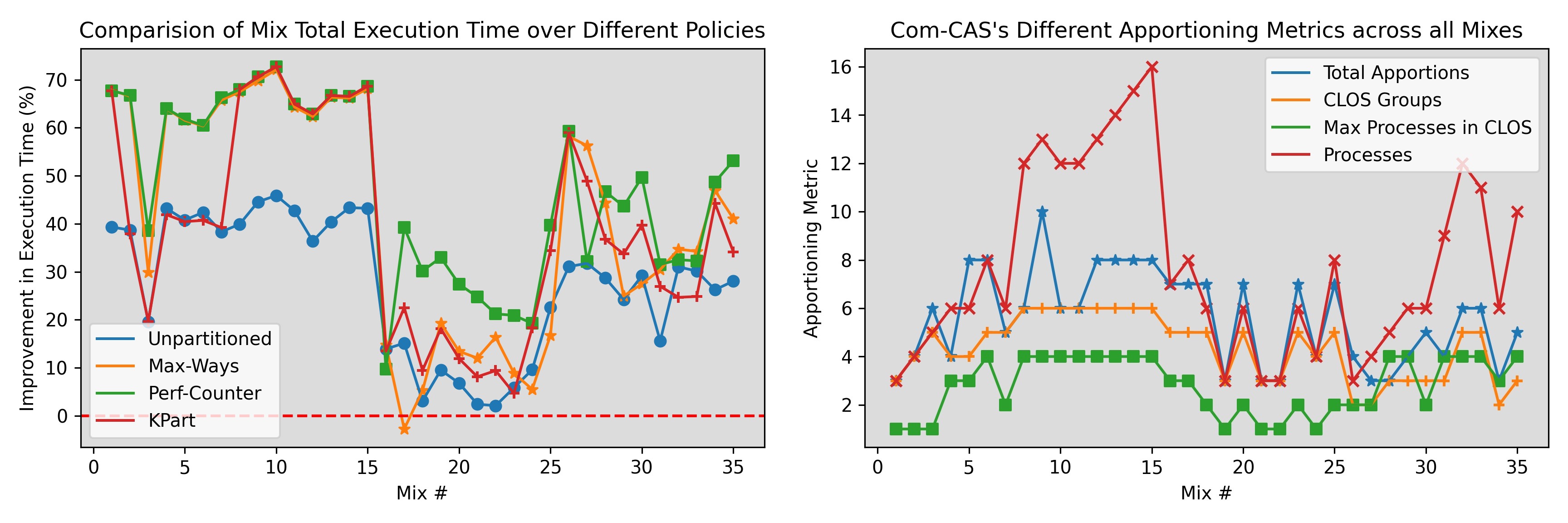}}
\caption{Total Mix Execution time and different Apportioning Metrics across all 35 mixes. The system throughput is determined by the total mix time. The different apportioning metrics give an overview of Com-CAS's decision making.}
\label{mix_time}
\end{figure*}


\subsection{Case Study: Closer look at Com-CAS}
\label{res:case}
In order to understand the decision making process of \textit{Com-CAS}, we present a case study comprising of a mix (\textit{\#16}) from GAPS, consisting of \textit{BC}, \textit{CC}, \textit{CC\_SV}, \textit{TC}, \textit{PR}, \textit{SSSP} \& \textit{BFS}. The loop attributes (loop-timing, memory footprint, performance-sensitivity factor ($\alpha$), max-ways) are estimated by Probes Compiler Framework and are encapsulated in the probes outside each loop. For each of the 7 processes in the mix, the loop-timing and memory footprint prediction accuracies are illustrated in Fig. \ref{case_study}, along with the values of $\alpha$ \& Max-Ways.

\begin{figure*}[ht]
\centerline{\includegraphics[width=1.0\linewidth]{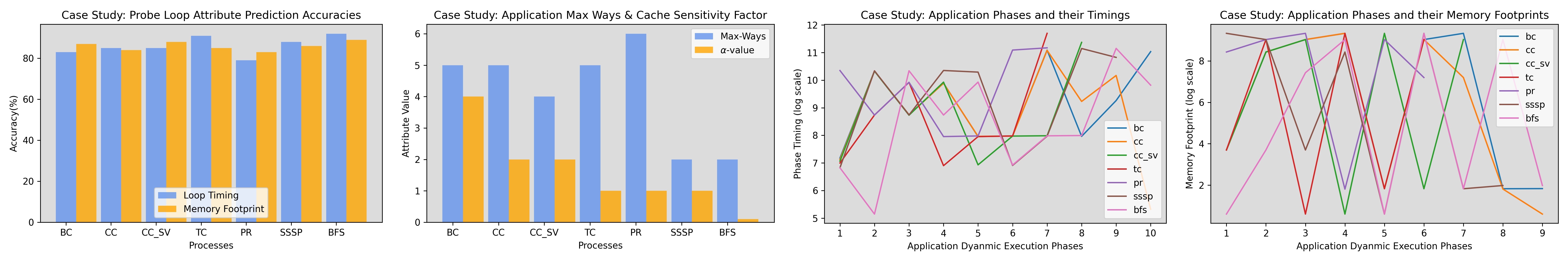}}
\caption{Case Study for 7-process mix (BC, CC, SSSP, TC, CC\_SV, BFS, PR) from GAPS: Probe Prediction Accuracies (left) and mix runtime behaviour (right). The rapid transitioning in log scale for phase-timing \& memory footprint illustrate the complex dynamic phase behaviour for this workload.}
\label{case_study}
\end{figure*}

During the execution, BCache Allocation Framework grouped \textit{BC} \& \textit{CC} in same socket because of their high $\alpha$ \& larger max-way values, while the rest five processes were put in another socket. At the start of the execution, \textit{BC}, \textit{CC}, \textit{CC\_SV}, \textit{TC}, \& \textit{PR}, were in reuse phase and exhibited footprints in the range of $10^3$ units, while \textit{SSSP} exhibits an initial footprint $\sim 10^9$ units, with streaming phase. \textit{BFS}'s initial phase footprint was only 4 units. Based on the footprint values (which are scaled by reuse behavior), $\alpha$, max-ways and phase-timing, processes \textit{BC}, \textit{CC}, \textit{CC\_SV}, \textit{TC}, \textit{PR} and \textit{SSSP} were initially allocated with 2 cache ways each, while \textit{BFS} was allocated 1 cache way. The interesting aspect of \textit{Com-CAS}'s apportioning is shown in the fact that although \textit{SSSP}'s footprint were considerably greater than other 5 process' footprints, its allocation was limited by its streaming behavior and max-ways value. 
Fig. \ref{case_study} shows that as the execution goes on further, processes \textit{BC} \& \textit{CC} undergoes through 10 dynamic phases, while there were 9 phase changes for \textit{BFS}, 8 each for \textit{CC\_SV} \& \textit{SSSP}, and 7 each for \textit{TC} \& \textit{PR}. These phases varied from an interval of $10^5$ to $10^{12}$ ns, while footprints changed from values as low as $~1-2$ units (negligible) to extremely high values such as $10^9$ units. The allocations for \textit{BC} changed from 2 to 4 (during phase $1\rightarrow2$) and then was shrunk back to 2 (during phase $7\rightarrow8$). Similarly, \textit{CC}'s ways were increased from 2 to 4 (during phase $2\rightarrow3$ ) and then was decreased back to 2 during the next phase transition. \textit {TC}'s ways were shrunk from 2 to 1 during phases $2\rightarrow3$ and then increased back to 2. All the other processes were either increased to their max-ways and then decreased back in subsequent phases or were kept constant throughout their execution ($SSSP$ \& $PR$). As we can see from Fig. \ref{case_study}, the application dynamic phase-behaviour are rapid and can transition from high footprint phase to low footprint one with varying reuse behaviour. Thus, the information broadcast from the probes helps us to make appropriate apportioning decisions.

\subsection{Effect on Cache Misses}
\label{res:cache}
\textit{Com-CAS}'s allocation algorithms focuses on enhancing the overall system performance by apportioning higher amount of cache to reuse-based processes that exhibit higher cache-sensitivity (high $\alpha$). As a result, the general trend observed is that reduction in LLC cache misses are shifted towards reuse-based applications that \textit{``need a greater amount of cache"}. This is important since the system throughput is determined by the execution time of process that has the longest execution time in mix and typically, they demand more cache and are cache-sensitive. Over the unpartitioned cache system, the processes \textit{Floyd-Warshall} (Polybench), \textit{BC} (GAP) \& \textit{SRad} (Rodinia) showed a reduction of \textbf{27.56\%}, \textbf{3.44\%} \& \textbf{44.9\%} respectively. For other applications present in the mix, the cache misses are either the same or have somewhat increased. This is because \textit{Com-CAS} prioritizes applications which exhibits higher degrees of cache-sensitivity. Also, for non-cache sensitive applications that have been instrumented with probes exhibit more misses because of the additional probe functions calls. Overall, \textit{Com-CAS} achieved a reduction of \textbf{4.6\%} in LLC misses over all the mixes in the unparitioned cache baseline. 

\subsection{SLA, Individual Process Latencies \& Fairness}
\label{sub:SLA}


\textit{Com-CAS} strives to enhance the overall throughput enhancement by catering to the requirements of processes that exhibit higher cache-sensitivity (high $\alpha$). However, following such a scheme can be detrimental to individual processes in the mix that are ``less-cache sensitive'' (low $\alpha$) or have minimal phase-timings (low $t$). Therefore, in addition to enhancing system throughput, it is important to focus on the individual process latencies in each of the workloads as well. For our experiments, we assume that \textbf{15\%} degradation compared to the `\textit{original-unmixed}' time, serves as an \textbf{acceptable latency degradation limit} (service-level-agreement/SLA). The original-unmixed time is measured by running the process individually in an unpartitioned cache system. In our experiments, we found that all 35 mixes adhere by this limit, i.e \textit{the individual process latency degradation is no more than 15\%}. To take a closer look, we choose one `representative mix' from each benchmark - the mix that obtained maximum throughput gains (\textit{best-performing mix}) over unpartitioned cache. Fig. \ref{poly_sla} compares the individual process times running as a mix with their \textit{original-unmixed} time in all 3 representative mixes and it's clear that none of the process show a latency degradation of more than 15\%. Note that in certain workloads with fully regular loops, the individual latencies are better than the original unmixed time; this achieved by the additional optimizations enabled by hoisting the probes to the outermost loop preheader. This effect is only noticeable in regular workloads with affine loops. Overall, across all 271 individual processes distributed in 35 mixes, \textit{Com-CAS} achieves an average performance degradation of \textbf{1.13\%}, compared to their original-unmixed time, and {\bf all} individual process latencies were within the SLA limit.

\begin{figure*}[!htbp]
\centerline{\includegraphics[width=1.0\linewidth]{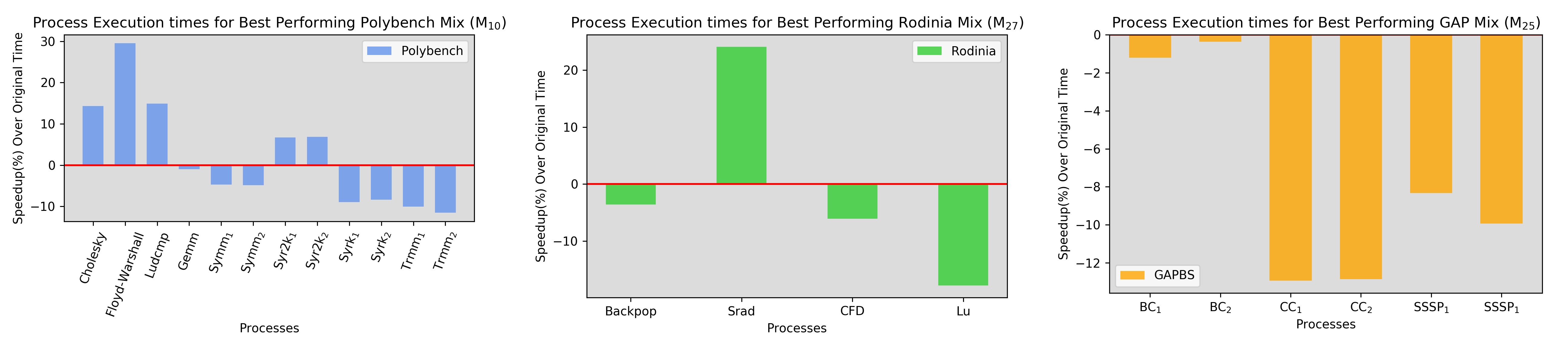}}
\caption{Comparison of individual process latencies with \textit{original-unmixed} time in representative mixes from each benchmark (\#10, \#25, \#27). All processes adhere to 15\% acceptable latency degradation limit (SLA).}
\label{poly_sla}
\end{figure*}

In addition to this, Fig. \ref{fairness} shows a comparison between \textit{Com-CAS} and the three other baseline partitioning policies - \textit{Max-Ways}, \textit{Perf-Counter}, and \textit{KPart}, on average. The fairness metric used here is \textbf{Jain's Fairness Index} \cite{jain1999throughput}, which is measured by $\frac{1}{1+\text{covariance}^2}$, where \textit{covariance} is obtained by $\frac{\sigma_{\text{mix-thoughput}}}{\mu_{\text{mix-thoughput}}}$. Fig. \ref{fairness} (left) shows the fairness index comparison in terms of mix size, which depicts that \textit{Com-CAS} maintains fairness in enhancing performance even for large mixes. This fact is further ratified by Fig. \ref{fairness} (right), where the fairness is compared with mixes having varying degrees of cache sensitivity ($\alpha_{mix}$ is defined as average $\alpha$ of all the processes in the mix) and \textit{Com-CAS} maintains an uniform fairness index of $\sim 0.99$ for different mix-categories.  

\begin{figure}[!ht]
\centerline{\includegraphics[width=1.0\linewidth]{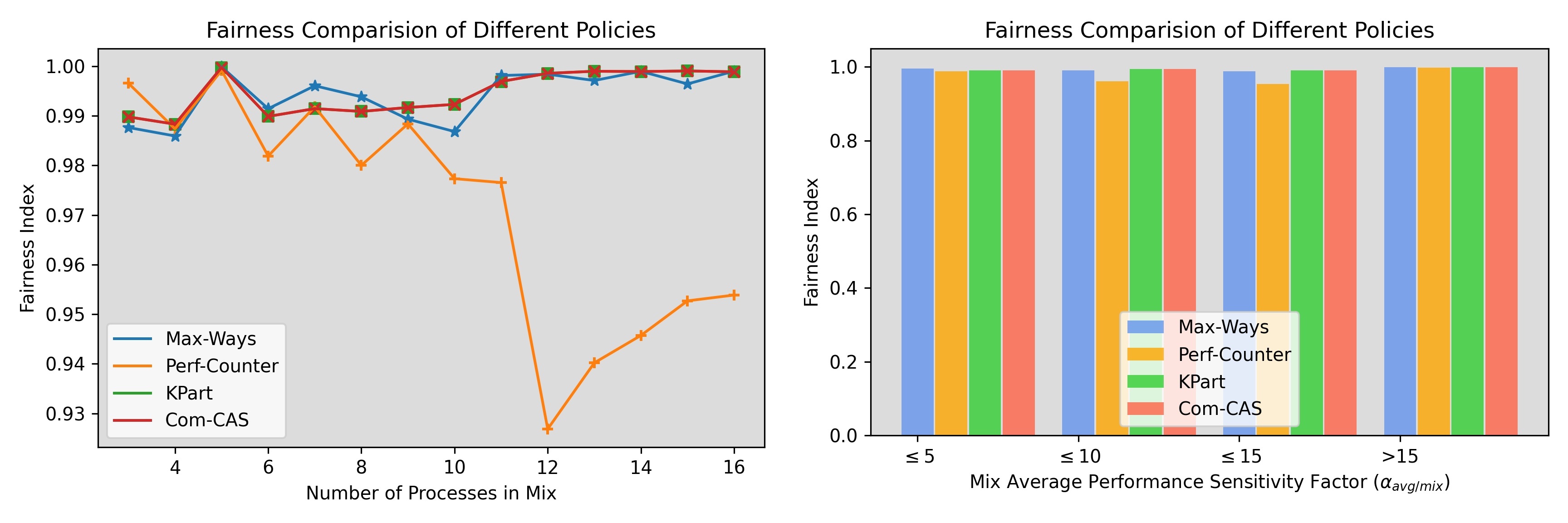}}
\caption{Fairness Comparison of Com-CAS with three partitioning policies as a function of total process and average mix cache-sensitivity}
\label{fairness}
\end{figure}

%% file: texes/background.tex
\section{Related Work}
\label{sec:background}

\textbf{Throughput-driven Approaches}: Recent works have focused on approaches to obtain finer-grain cache partitions \cite{wang2017swap,el2018kpart,pons2020phase}. KPart \cite{el2018kpart} (\S\ref{sec:res}) takes a clustering approach that requires profiling every pair of applications in the mix in addition to profiling them separately, i.e performing $O(n^2)$ profiling operations on top of $O(n)$ individual operations. Apart from being non-scalable for mixes containing large number of applications, this approach is not practical in real-world scenarios, where mix composition or even participating processes therein are not even known apriori. Similarly, another recent work \cite{xiang2018dcaps} proposes an approach of online monitoring using miss-curves also suffers from the same drawbacks - the fact that workloads might be input sensitive and the miss-curves might vary according to the execution scenario is ignored. An alternative approach \cite{pons2020phase} that accounts for application phase changes, where based on the execution behavior, applications are grouped into multiple categories by extensive profiling. However, such static categorization misses an important aspect - application phases are input-dependent and are dynamic as well as contingent upon entire mix composition. Thus, an application behaviour might change based on an different execution scenario or with new input data. Since this approach misses accounting for dynamic interactions and concurrency of different application phases, we feel that application classification is not a feasible approach. Also, the evaluations are primarily done on \textit{SPEC CPU} benchmarks \cite{henning2006spec}, which has been shown to be non-cache sensitive in \cite{cook2013hardware}.

\textbf{Latency-driven approaches}: Our work primarily focuses on \textit{throughput improvement} and the 15\% SLA requirement serves as an upper bound for \textit{allowable latency degradation}. Several resource managing mechanisms \cite{chen2019parties,zhu2016dirigent, lo2015heracles,roy2021satori,xiang2019emba,kasture2014ubik, kasture2015rubik, 9470449} use CAT to perform dynamic resource allocations for latency-critical server-based applications; these approaches also focus on many critical resources such as network bandwidths, memory hierarchies etc, which are critical in determining QoS guarantees, maximization of utilization and effective co-location. In contrast, our focus is on developing smart apportioning techniques for cache-sensitive applications for throughput enhancement while maintaining isolation. Furthermore, majority of these approaches have several drawbacks. For instance, \textit{PARTIES} \cite{chen2019parties} monitors behaviours of latency-critical applications every $~10^2ms$, and takes corrective actions only \textit{after} QoS violation has been detected. This makes it susceptible to detection lag and also ill-equipped to handle rapid phase-changes that are in order of few ms (Fig. \ref{case_study}). A multi-resource partitioning scheme \cite{roy2021satori}, constructs bayesian curves to find ``near-optimal'' resource configuration. It's agnostic to the fact that applications batches may not remain stable in an actual multi-tenant execution setting, and could be continuously updated (bayesian optimization solution can change). Additionally, all these approaches have been evaluated with smaller mix-sizes. On the other hand, \textit{Com-CAS} can be utilized in system configurations that execute data-intensive and cache-sensitive applications. Workloads in such scenarios are mostly optimized and tuned for in-core memory execution, and thus are agnostic to network connections, remote data etc.

%% file: texes/conclusion.tex
\section{Conclusion}
\label{sec:con}
In this work, we proposed \textbf{Compiler-Guided Cache Apportioning System (Com-CAS)} for effectively apportioning the shared LLC leveraging Intel CAT. \textbf{Probes Compiler Framework} evaluates cache-atributes such as reuse behaviour, cache footprints, loop timings and cache sensitivity and relays them. Using these information, \textbf{BCache Allocation Framework} uses allocation algorithms to dynamically partitions the cache and schedules processes. \textit{Com-CAS} improved average throughput by \textbf{15\%} on unparititioned cache system, and \textbf{20}\% on state-of-art KPart~\cite{el2018kpart}, while maintaining the worst individual application execution time degradation within 15\% in a multi-tenancy setting. In addition, \textit{Com-CAS}'s scheduling minimizes the co-location of reuse applications and curtails their overlap. With improved throughput, fulfilled SLA agreement and reduced process interference, we contend that the proposed \textit{Com-CAS} is a viable system for multi-tenant setting.